\documentclass[conference,letterpaper,10.5pt]{IEEEtran}

\usepackage[utf8]{inputenc} 
\usepackage[T1]{fontenc}
\usepackage{url}
\usepackage{ifthen}
\usepackage{cite}
\usepackage[cmex10]{amsmath} 

\usepackage{amsfonts, amsthm}
\usepackage{ dsfont }
\renewcommand{\r}{\boldsymbol{r}}
\newcommand{\x}{\boldsymbol{x}}
\newcommand{\z}{\boldsymbol{z}}
\renewcommand{\c}{\boldsymbol{c}}

\newcommand{\y}{\boldsymbol{y}}

\newcommand{\Bin}{Bin}

\newtheorem{theorem}{Theorem}
\newtheorem{proposition}{Proposition}

\newtheorem{definition}{Definition}

\title{Fast Decoding of Union-Free Codes}
\author{Ilya Vorobyev \\
	Skolkovo Institute of Science and Technology
}

\date{}

\begin{document}

	

	\maketitle
	
	\begin{abstract}
		Union-free codes and disjunctive codes are two combinatorial structures, which are used in a non-adaptive group testing to find a set of $d$ defective elements among $n$ samples by carrying out the minimal number of tests $t$.  It is known that union-free codes have a larger rate, whereas disjunctive codes provide a more efficient decoding algorithm. In this paper we introduce a new family of codes for a non-adaptive group testing with fast decoding. The rate of these codes is not less than the rate of disjunctive codes, while the decoding algorithm has the same complexity. In addition, we derive a lower bound on the rate of new codes for the case of $d=2$ defectives, which is significantly better than the bound for disjunctive codes and almost as good as the bound for union-free codes.
		
	\end{abstract}

	\section{Introduction}\label{sec::intro}
	Group testing is a well-known combinatorial problem, which was introduced by Dorfman in 1943~\cite{dorfman1943detection}. Suppose that we have a large set of $n$ samples, some of which are defective. Our goal is to find the set of defective samples by using the minimal number of tests on properly chosen subsets of defectives. The result of the test is positive if at least one tested sample is defective and negative otherwise.
	
	There are two main types of group testing algorithms. In \textit{adaptive} algorithm tests are performed one by one; the results of previous tests can be used to design the next test. In \textit{non-adaptive} algorithms all tests are predefined and can be performed in parallel. 
	
	On the other hand, group testing problems can be divided into combinatorial and probabilistic models. In the combinatorial model we design algorithms, which always find the set of defectives, provided that its cardinality is at most $d$. In the probabilistic model there is a distribution over the sets of defectives. The designed algorithm should be able to find the set of defectives with high probability. 
	
	In this paper we consider only combinatorial non-adaptive algorithms. We refer an interested reader to~\cite{du2000combinatorial,aldridge2019group,d2014lectures,cicalese2013fault}, where other models are described.
	
	In a seminal paper~\cite{kautz1964nonrandom} the authors introduced two families of codes, which can solve non-adaptive group testing problems. \textit{Union-free codes} guarantee that any two distinct sets of defectives give different vectors of test results. Therefore, it is possible to restore the set of defectives from the vector of test results. However, the brute force decoding algorithm, which iterates over all possible sets of $\leq d$ defectives, requires $O(tn^d)$ operations in the worst case, where $t$ is a number of tests. \textit{Disjunctive codes} (also known as \textit{cover-free codes} in the literature) provide a much more efficient decoding algorithm, which requires only $O(tn)$ operations. Yet, disjunctive codes have a lower rate, i.e. for the same number of samples we need to perform a larger number of tests to find the set of defectives.
	
	In a recent paper~\cite{fan2021strongly} the authors defined a new notion of  \textit{a strongly $d$-separable matrix} ($d$-SSM), which also can be used to find the set of $d$ defectives. It is proved that the rate of codes from this family is not less than the rate of disjunctive codes, whereas it provides an efficient $O(tn)$ decoding algorithm. Moreover, the authors of~\cite{fan2021strongly} proved a lower bound $0.2213$ on the rate of $2$-SSM, which is better than the best lower bound $0.1814$ on the rate of $2$-disjunctive codes.
	
	In our paper we modify union-free codes to provide them with an efficient $O(tn)$ decoding algorithm. The rate of the new family of codes is not less than the rate of disjunctive codes. For the case of $d=2$ defectives we prove a lower bound $0.3017$ on the rate, which is significantly better than the bound for $2$-SSM. Note that it is only slightly less than the lower bound $0.3135$ for union-free codes, proved in~\cite{coppersmith1998new}.
	
	The rest of the paper is organized as follows. In Section~\ref{sec::notation} we give all necessary notations and discuss some known results. In Section~\ref{sec::main} we introduce a new family of codes and its decoding algorithm. The main result is proved in Section~\ref{sec::proof}. Section~\ref{sec::conclusion} concludes the paper.
	
	\section{Preliminaries}\label{sec::notation}
	Denote the set of integers $\{1,2,\ldots,n\}$ as $[n]$. A number of non-zero entries in a vector is called its \textit{weight}.
	Call a vector $\z=(\z(1),\ldots,\z(n))$ a Boolean sum of vectors $\x=(\x(1), \ldots, \x(n))$ and $\y=(\y(1), \ldots, \y(n))$ if $\z(i)=\x(i)\vee\y(i)$. Denote such operation as $\z=\x\bigvee\y$. Say that a vector $\x$ covers a vector $\y$ if $\x\bigvee\y=\x$. The binary entropy function $h(x)$ is defined as follows
	$$
	h(x)=-x\log_2x-(1-x)\log_2(1-x).
	$$ 
	
	A non-adaptive algorithm consisting of $t$ tests can be described with a matrix $C$ of size $t\times n$. 
	Each row corresponds to one test and each column corresponds to one sample. We put $1$ into intersection of a column $i$ and a row $j$ if the sample $i$ is included in the test $j$. 
	We refer to columns of the matrix $C$ as $\c_1,\ldots, \c_n$. A vector $\r$ of testing outcomes can be expressed as follows
	$$
	\r=\bigvee\limits_{i\in D}\c_i,
	$$
	where $D$ is the set of defectives. If $D$ is an empty set then we define $\r$ as all-zero vector.
	
	Introduce the definitions of union-free codes and disjunctive codes, which were called \textit{uniquely decipherable} and \textit{zero-drop-rate} codes in the original work~\cite{kautz1964nonrandom}.
	
	\begin{definition}\label{def::union-free}
		Call a matrix $C$ a $d$-union-free code if for any two different sets $D_1, D_2\subset [n]$, $|D_1|,|D_2|\le d$, the Boolean sums $\bigvee\limits_{i\in D_1}\c_i$ and $\bigvee\limits_{i\in D_2}\c_i$ are distinct.
	\end{definition}
	
	\begin{definition}\label{def::disjunctive}
		Call a matrix $C$ a $d$-disjunctive code if for any  set $D\subset [n]$, $|D|=d$, the Boolean sum $\bigvee\limits_{i\in D}\c_i$ doesn't cover any column $\c_j$ for $j\in [n]\setminus D$.
	\end{definition}
	
	The well-known connection between these two families of codes was established in~\cite{kautz1964nonrandom}.
	\begin{proposition}
		A $d$-disjunctive code $C$ is a $d$-union-free code. A $d$-union-free code is a $(d-1)$-disjunctive code.
	\end{proposition}
	
	Both union-free codes and disjunctive codes can be used in combinatorial group testing to identify the set of $\le d$ defectives. Indeed, by definition of union-free codes there exist exactly one set $D$ of cardinality at most $d$, which can produce the observed vector of outcomes $\r$. However, to find this set $D$ we need to iterate over all possible sets of cardinality at most $d$, which requires $O(tn^d)$ operations. 
	
	Since disjunctive codes are also union-free codes, we know that they can be used in combinatorial group testing. Moreover, the decoding can be done efficiently using $O(tn)$ operations. The decoding algorithm outputs all indexes $i$, such that column $\c_i$ is covered by the vector of outcomes $r$. In other words, the algorithm outputs all samples, which appear only in positive tests. It is clear that all defectives will be found and no other elements will be included in the answer. In probabilistic group testing this decoding algorithm is known under the name COMP (Combinatorial orthogonal matching pursuit)~\cite{aldridge2019group}.
	
	Let's recall the definition of $d$-SSM codes from~\cite{fan2021strongly}.
	\begin{definition}
		A binary $t\times n$ matrix $C$ is called a strongly $d$-separable matrix (code), if for any $D_0\subset [n]$, $|D_0|=d$, we have
		$$
		\bigcap\limits_{D'\in U(D_0)}D'=D_0,
		$$
		where
		$$
		U(D_0)=\left\{D'\subset [n]:\bigvee\limits_{i\in D'} \c_i=\bigvee\limits_{i\in D_0} \c_i\right\}.
		$$
	\end{definition}

	It was established that $d$-SSM codes lie in between $d$-union-free codes and $d$-disjunctive codes.
	\begin{proposition}[Lemma II.3 from~\cite{fan2021strongly}]
		A $d$-disjunctive code is a $d$-SSM code. A $d$-SSM code is a $d$-union-free code.
	\end{proposition}
	
	The authors of the paper~\cite{fan2021strongly} proposed an efficient $O(tn)$ decoding algorithm for $d$-SSM code. In fact, their algorithm coincides with DD (definitely defective) algorithm, which has been previously used in probabilistic group testing (see~\cite{aldridge2014group,aldridge2019group}).

	Let $n_{UF}(t, d)$, $n_{D}(t, d)$ and $n_{SSM}(t, d)$ denote the maximum cardinality of $d$-union-free, $d$-disjunctive and $d$-SSM codes with length $t$. Introduce their maximum asymptotic rates
	$$
	R_{UF}(d)=\limsup_{t\to\infty}\frac{\log_2 n_{UF}(t, d)}{t};
	$$
	$$
	R_{D}(d)=\limsup_{t\to\infty}\frac{\log_2 n_{D}(t, d)}{t};
	$$
	$$
	R_{SSM}(d)=\limsup_{t\to\infty}\frac{\log_2 n_{SSM}(t, d)}{t}.
	$$
	
	Obviously,
	$$
	R_{D}(d)\leq R_{SSM}(d)\leq R_{UF}(d).
	$$	
	Therefore, $d$-SSM codes have the largest rate among the known codes with efficient decoding algorithm.
	
	\section{Union-free Codes with Fast Decoding}\label{sec::main}
	In this section we define union-free codes with fast decoding (UFFD codes) and describe their decoding algorithm.
	\begin{definition}
		A $t\times n$ binary matrix $C$ is a $d$-union-free code with fast decoding (UFFD code), if 
		\begin{enumerate}
			\item It is a $d$-union-free code.
			\item Each possible vector of outcomes $\r$, i.e. a vector, which can be obtained as a Boolean sum of at most $d$ columns of $C$, covers at most $n^{1/d}$ columns from $C$.
		\end{enumerate}
	\end{definition}
	The relation of UFFD codes with union-free and disjunctive codes is formulated in the following obvious proposition.
	
	\begin{proposition}\label{pr::UFFD and other codes}
		A $d$-disjunctive code is a $d$-UFFD code. A $d$-UFFD code is a $d$-union-free code.
	\end{proposition}
	
	Let $n_{UFFD}(t, d)$ be the maximum cardinality of $d$-UFFD code of length $t$. Define the maximum asymptotic rate
	$$
	R_{UFFD}(d)=\limsup_{t\to\infty}\frac{\log_2 n_{UFFD}(t, d)}{t}.
	$$
	Proposition~\ref{pr::UFFD and other codes} implies
	$$
	R_{D}(d)\leq R_{UFFD}(d)\leq R_{UF}(d).
	$$
	
	Now we describe a decoding algorithm of $d$-UFFD code. The algorithm consists of two steps.
	
	\textbf{Step 1.} Find the set $\hat{D}$ of all indexes $i$, such that $\c_i$ is covered by the vector of outcomes $\r$.
	
	\textbf{Step 2.} Iterate over all sets $D_0$,  $D_0\subset \hat{D}$,   $|D_0|\leq d$. If $\bigvee\limits_{i\in D_0}\c_i=\r$, then output $D_0$ as an answer and finish the algorithm.
	
	\begin{proposition}
		The described algorithm successfully finds the set of $\leq d$ defectives. The number of required operations is $O(tn)$.
	\end{proposition}
	\begin{proof}
		Since UFFD codes are union-free codes, we know that there exists exactly one set $D_0$,  $D_0\subset \hat{D}$,   $|D_0|\leq d$ such that $\bigvee\limits_{i\in D_0}\c_i=\r$. This set will be found by the algorithm during the second step.
		
		To perform the first step we need to scan the matrix $C$, so its complexity $O(tn)$. Complexity of the second step can be upper bounded by $O(t|\hat{D}|^d) $. Definition of UFFD codes guarantees us that $|\hat{D}|\leq n^{1/d}$, therefore the total complexity is $O(tn)$.
		
	\end{proof}

	\section{Lower bound for $d=2$}\label{sec::proof}
	Let's recall known bounds for different families of codes.
	For $d\to\infty$ lower~\cite{dyachkov1989superimposed,quang1988bounds} and upper~\cite{d1982bounds,ruszinko1994upper,furedi1996onr} bounds look as follows
	\begin{align*}
		\frac{\ln 2}{d^2}(1+o(1))\leq 
		R_{D}(d)\leq R_{SSM}(d)\\
		\leq R_{UF}(d)\leq \frac{2\log_2d}{d^2}(1+o(1)).
	\end{align*}
	
	For $d=2$ it is known~\cite{coppersmith1998new,d1982bounds,erdos1982families,fan2021strongly}
	$$
	0.3135\leq R_{UF}(2)\leq 0.4998;
	$$
	$$
	0.1814\leq R_{D}(2)\leq 0.3219;
	$$
	$$
	0.2213\leq R_{SSM}(2)\leq 0.4998.
	$$
	
	From Proposition~\ref{pr::UFFD and other codes} we obtain
	$$
	\frac{\ln 2}{d^2}(1+o(1))\leq R_{UFFD}(d)\leq \frac{2\log_2d}{d^2}(1+o(1))
	$$
	for $d\to\infty$
	and 
	$$
	0.1814\leq R_{UFFD}(2)\leq 0.4998
	$$
	for $d=2$.
	
	In the following theorem we improve the lower bound for $d=2$.
	\begin{theorem}\label{th::lower bound}
		$$R_{UFFD}(2)\geq 0.3017.$$
		More precise,
		\begin{equation*}
			R_{UFFD}(2)\geq\max\limits_{p\in(0, 1)}\min(R_0, R_1), 
		\end{equation*}
		where
		\begin{equation*}
			R_0=\min\limits_{\alpha\in A}
			\frac{4h(p)-2\alpha h(p/\alpha)-2ph((\alpha-p)/p)-h(\alpha)}{3},
		\end{equation*}
		
		\begin{equation*}
			R_1= \min\limits_{\alpha\in A}\frac{2h(p)-2ph((\alpha-p)/p)-(1-p)h(\frac{\alpha-p}{1-p})}{2}
		\end{equation*}
		and $A=(p, \min(2p, 1)).$
	\end{theorem}
	
	Note that this bound is significantly larger than lower bounds for $R_D(2)$ and $R_{SSM}(2)$.
	\begin{proof}
		Consider a random binary $t\times n$ matrix $C$, $n=\lfloor 2^{Rt}\rfloor$, each column of which is chosen independently and equiprobable from the set of column of a fixed weight $\lfloor pt\rfloor$ for some real $p\in(0,\,1)$. Parameters $p$ and $R$ would be specified later. In the following we omit the $\lfloor \cdot\rfloor$ sign for simplicity, since it doesn't affect the asymptotic rate. Fix two sets of indexes $D_1$ and $D_2$, $D_1, D_2\subset [n]$, $|D_1|=|D_2|=2$. Call a pair of sets $D_1$ and $D_2$ bad, if 
		\begin{equation}\label{eq::equal outcomes}
			\bigvee\limits_{i\in D_1}\c_i=\bigvee\limits_{i\in D_2}\c_i.
		\end{equation}
		
		Let's estimate a mathematical expectation $E$ of  a number of bad pairs. Consider two cases: $|D_1\cap D_2|=0$ and $|D_1\cap D_2|=1$.
		
		The probability of~\eqref{eq::equal outcomes} for the case $|D_1\cap D_2|=0$ is equal to
		\begin{equation}
			P_0=\sum\limits_{k=pt}^{\min(2pt,t)} \frac{\binom{t}{k}\binom{k}{pt}^2\binom{pt}{k-pt}^2}{\binom{t}{pt}^4},
		\end{equation}
		where $k$ is a weight of a possible outcome vector $\r$. The mathematical expectation of the number of such bad pairs is upper bounded as follows
		$$ E_0\leq n^4P_0.$$
		
		For the case $|D_1\cap D_2|=1$ the probability of~\eqref{eq::equal outcomes} is equal to
		\begin{equation}
			P_1=\sum\limits_{k=pt}^{\min(2pt, t)}\frac{\binom{t-pt}{k-pt}\binom{pt}{k-pt}^2}{\binom{t}{pt}^2},
		\end{equation}
		and the mathematical expectation is at most
		$$
		E_1\leq n^3P_1.
		$$
		
		Let $R$ be the maximal value such that $E_0<n/8$ and $E_1<n/8$. For such $R$ the mathematical expectation of  the number of bad pairs is less than $n/4$, therefore, by Markov's inequality the probability that the number of bad pairs greater than $n/2$ is less than $1/2$.
		
		Let's find this parameter $R$. In the following inequalities we use the fact $\binom{n}{k}=2^{n(h(k/n)+o(1))}$, $n\to\infty$.
		
		\begin{align*}
			\log_2 E_0&\leq 4\log_2 n + \log_2 P_0\\
			&\leq
			4\log_2 n + \log_2 \left(t \max_{k\in [pt, \min(2pt, t)]}  \frac{\binom{t}{k}\binom{k}{pt}^2\binom{pt}{k-pt}^2}{\binom{t}{pt}^4}\right)\\
			&=4\log_2 n + o(\log_2 n) +t\\&\times\max_{\alpha\in A}\left(h(\alpha)+2\alpha h(\frac{p}{\alpha})+2ph(\frac{\alpha-p}{p})-4h(p)\right)
		\end{align*}
	    This expression should be less than $\log_2(n/8)=\log_2 n+o(\log_2 n)$, which is equivalent to
	    $$
	    \frac{\log_2n}{t}\leq R_0+o(1).
	    $$
	    
	    Now we deal with $E_1$. 
	    \begin{align*}
	    	\log_2 E_1&\leq 3\log_2 n + \log_2 P_1\\
	    	&\leq
	    	3\log_2 n + \log_2 \left(t \max_{k\in [pt, \min(2pt, t)]}  \frac{\binom{t-pt}{k-pt}\binom{pt}{k-pt}^2}{\binom{t}{pt}^2}\right)\\
	    	&=3\log_2 n + o(\log_2 n) +t
	    	\\&\times\max_{\alpha\in A}\left(
	    	(1-p)h(\frac{\alpha-p}{1-p})+2ph(\frac{\alpha-p}{p})-2h(p)\right)
	    \end{align*}
    
    The last expression should be less than $\log_2(n/8)=\log_2 n+o(\log_2 n)$, which is equivalent to
    $$
    \frac{\log_2n}{t}\leq R_1+o(1).
    $$
		
		So, we obtain that
		\begin{equation*}
			R\geq\max_p\min(R_0, R_1)+o(1). 
		\end{equation*}

		Numerical computations show that the optimal value $R$ is attained in $p\approx 0.3105$ and equals $\approx0.3017$.
		
		Now estimate the probability $P$ that one fixed response vector of weight $\leq 2pt$ covers more then $\sqrt{n/2}$ columns of $C$. The number of covered columns is stochastically dominated by a binomial random variable $\xi\sim\Bin(n, q) $, where
		$$
		q=\frac{\binom{2pt}{pt}}{\binom{t}{pt}}=2^{(2p-h(p)+o(1))t}<n^{-0.9+o(1)}
		$$
		for $p=0.3105$.
		Then the probability $P$ is upper bounded as follows
		\begin{align*}
			P&\leq n\binom{n}{\sqrt{n/2}} \left(n^{-0.9+o(1)}\right)^{\sqrt{n/2}}\\
			&<
			\frac{n^{(0.1+o(1))\sqrt{n/2}}}{\sqrt{n/2}!}<n^{(-0.4+o(1))\sqrt{n/2}}.
		\end{align*}
		The probability that at least one response vector of weight $\leq 2pt$ covers more then $\sqrt{n/2}$ columns can be upper bounded as 
		$2^tP\to0$.
		
		We conclude that for any $t$ big enough with a positive probability there exist a $t\times n$ code $C$ with a rate $R\geq 0.3017+o(1)$ such that the mathematical expectation of the number of bad pairs is at most $n/2$, and the number of columns, which can be covered by some vector of weight $\leq 2pt$, is at most $\sqrt{n/2}$. After deleting of one column from every bad pair we obtain a 2-UFFD code of length $t$, cardinality $n/2$ and rate $0.3017+o(1)$.
		
	\end{proof}
	\section{Conclusion}\label{sec::conclusion}
	In this paper we introduced a new family of codes, which solves the combinatorial group testing problem and provides an efficient $O(tn)$ decoding algorithm. The rate of the new family lies  between the rates of union-free and disjunctive codes. For the case of $d=2$ defectives we proved a new lower bound on the rate $0.3017$, which is significantly better than corresponding bounds for disjunctive codes and strongly separable matrices. An interesting open question is to explore the possibility of the application of the new decoding algorithm in other group testing problems.
	\section*{Acknowledgement}
	The reported study was supported by RFBR and JSPS under Grant No.~20-51-50007, and by RFBR and National Science Foundation of Bulgaria (NSFB), project number 20-51-18002. 
	\bibliographystyle{IEEEtran}
	\bibliography{gt}

	
	

\end{document}